\documentclass[pdftex,pra,aps,twocolumn,nofootinbib,showkeys,showpacs]{revtex4-1}
\usepackage{silence}
\usepackage{amsmath}
\usepackage{amsfonts}
\usepackage{amsthm}
\usepackage{graphicx}
\usepackage{caption}
\usepackage{subcaption}

\usepackage{enumitem}
\usepackage{tikz}
\usepackage{amsfonts}
\usepackage{amsmath,amssymb,amsthm}
\usepackage{graphicx}
\usepackage{fancyhdr}
\usepackage[breakable]{tcolorbox}
\usepackage{bbm}
\usepackage{url}

\usepackage{braket}
\usepackage{upgreek }
\usepackage{dsfont}

\usepackage[plain]{fancyref} 
\usepackage{algorithm}
\usepackage{algpseudocode}

\usepackage{hyperref}

\newtheorem{lemma}{Lemma}
\newtheorem{observation}{Observation}
\newtheorem{theorem}{Theorem}
\newtheorem{corollary}{Corollary}

\DeclareMathOperator{\poly}{poly}

\begin{document}

\title{Stoquastic ground states are classical thermal distributions}
\author{Robbie King, Sergii Strelchuk}
\affiliation{DAMTP, Centre for Mathematical Sciences, University of Cambridge, Cambridge~CB30WA, UK}

\begin{abstract}
We study the structure of the ground states of local stoquastic Hamiltonians and show that under mild assumptions the following distributions can efficiently approximate one another: (a) distributions arising from ground states of stoquastic Hamiltonians, (b) distributions arising from ground states of stoquastic frustration-free Hamiltonians, (c) Gibbs distributions of local classical Hamiltonian, and (d) distributions represented by real-valued Deep Boltzmann machines. In addition, we highlight regimes where it is possible to efficiently classically sample from the above distributions.
 
\end{abstract}
\maketitle

\section{Introduction}
Understanding the properties of quantum many-body systems is one of the biggest challenges in modern physics with profound implications in solid-state physics, quantum chemistry and quantum computing. The behaviour of these systems is characterized by a Hamiltonian, and the simplest question one can ask is about the properties of its ground state space. The knowledge of the latter space allows us to encode a solution to a large number of computational problems in the ground state of a quantum (time-dependent) Hamiltonian, and utilize methods of adiabatic quantum computation~\cite{farhi2000quantum}.

Efforts to characterize the ground space structure of general quantum Hamiltonians, using some of the most powerful simulation methods such as quantum Monte Carlo, faces considerable difficulties due to the so-called `sign problem' -- a major obstacle which results in prohibitively long convergence times for the probabilistic algorithms.

The study of quantum Hamiltonians which do not exhibit sign-problem motivated the introduction of stoquastic Hamiltonians~\cite{bravyi2006complexity}. They have the property that the off-diagonal entries are nonpositive in the standard basis. This ensures that the ground state of the Hamiltonian has real nonnegative amplitudes when expressed in this basis.
Stoquastic Hamiltonians describe a wide range of physical systems: from ground states of the transverse field Ising model~\cite{bravyi2017complexity} to Jaynes-Cummings model~\cite{walls2007quantum} and certain classes of superconducting qubits~\cite{kjaergaard2020superconducting}. Moreover, stoquastic ground states arise in the study of an important class of projected entangled pair states~\cite{verstraete2006criticality}.

In the complexity-theoretic context, stoquastic Hamiltonians have been widely studied. They gave rise to a complexity class StoqMA which is contained in the QMA, the quantum analogue of the complexity class NP, and it is expected that this containment is strict~\cite{kjaergaard2020superconducting}. Despite being sign-problem free, stoquastic Hamiltonians occupy a curious intermediate position between classical and quantum Hamiltonians. 
In particular, if the stoquastic adiabatic computation is universal for quantum computing this would imply the collapse of the Polynomial Hierarchy~\cite{bravyi2006merlin}. Also, it is known that stoquastic Hamiltonians encompass classical Ising model for which estimating its ground state energy is known to be NP-hard. There also exist simple Hamiltonians whose action is restricted to two-qubit interactions that can simulate any stoquastic Hamiltonian~\cite{bravyi2017complexity, cubitt2018universal}. Recent works further investigated the complexity of deciding whether a given $2$-local Hamiltonian has a sign problem~\cite{klassen2019two} and the algorithmic difficulty of curing the latter~\cite{marvian2019computational}. An important class of stoquastic Hamiltonians has the frustration-free property: each of its terms is stoquastic and acting on a constant number of qubits and the ground state of the overall Hamiltonian minimizes the energy of each of its terms. It is known that adiabatic evolution of such Hamiltonians can be efficiently simulated on a classical probabilistic computer~\cite{bravyi2010complexity}.

In our work, we find a precise mathematical formalism which characterizes the ground states of the stoquastic Hamiltonians. We show that one can succinctly express the structure of the ground state space in terms of Boltzmann machines. This neural network formalism, originally inspired by ideas from statistical mechanics, represents a class of energy-based models that found numerous uses in physics describing spin glasses, Ising model~\cite{alberici2020annealing} and multiple machine-learning applications~\cite{salakhutdinov2009deep}. In recent years, it has been extended to describe quantum systems by quantum neural network states~\cite{carleo2017solving}, leading to a flurry of results in condensed matter physics~\cite{shi2019neural}, quantum error correction~\cite{torlai2017neural, li2019improved}, quantum computing and beyond~\cite{gao2017efficient}. There exist several variations of Boltzmann machines, depending on the underlying graph structure~\cite{salakhutdinov2009deep}. The simplest such machine is called the Restricted Boltzmann Machine, whose graph is comprised of one visible and one hidden with the latter having no connections between hidden units. They have been widely used as generative models in classical machine learning~\cite{hinton2002training,salakhutdinov2008quantitative} and more recently in quantum physics~\cite{melko2019restricted, vieijra2020restricted} due to the ease of training and sampling. However, not all quantum systems admit an efficient representation as Restricted Boltzmann Machines~\cite{gao2017efficient}. This is remedied by considering a richer model, Deep Boltzmann Machines, which has two or more hidden layers. While it can represent quantum states efficiently, training and sampling become prohibitively slow in general~\cite{gao2017efficient}. We construct a new type of Boltzmann machine (we call it Hyper Boltzmann Machine (HBM)) which gives rise to probability distributions that precisely capture the correlations of the stoquastic ground state space and investigate its properties. We show that its representational power is comparable to the Deep Boltzmann Machines, but it naturally captures the properties of probability distributions that can be encoded in the ground states of stoquastic Hamiltonians. 

More precisely, under mild assumptions, the following classes of probability distributions can efficiently approximate one another with a polynomial overhead in the size of the system: 
\begin{itemize}
	\item Distributions arising from ground states of local stoquastic Hamiltonians.
	\item Distributions arising from ground states of local stoquastic frustration-free (SFF) Hamiltonians.
	\item Gibbs distributions of local classical Hamiltonians.
	\item Distributions arising from Deep Boltzmann machines.
\end{itemize}
In addition, we highlight an explicit link between the Boltzmann machine formalism and the classical Ising model. Finally, we investigate regimes when these distributions become classically efficiently simulable.

\bigskip

\section{Preliminaries}

\subsection{Hamiltonians}

A $k$-local Hamiltonian $H$ on $n$ qubits takes the form $H = \sum_a H_a$, a sum over $O(\poly(n))$ terms, where each term $H_a$ acts non-trivially on at most $k$ qubits. A $k$-local Hamiltonian $H = \sum_a H_a$ is stoquastic if each $H_a$ has real non-positive off-diagonal entries in the computational basis $\langle x|H_a|y \rangle \leq 0 \ \forall \ x \neq y$. The ground state $|\psi \rangle$ of a stoquastic Hamiltonians can be taken to have real positive amplitudes in the computational basis $\langle x|\psi \rangle \geq 0 \ \forall \ x$. A $k$-local Hamiltonian $H = \sum_a H_a$ is frustration free if the ground state $|\psi \rangle$ of $H$ is simultaneously a ground state of each individual term $H_a$. For each classical (i.e. diagonal) Hamiltonian we associate the Gibbs distribution (at temperature 1) $p(x) = \frac{1}{Z} e^{-H(x)}$, where $Z$ is a normalizing constant (the partition function). The corresponding coherent Gibbs state is the quantum state $|\psi \rangle = \frac{1}{\sqrt{Z}} \sum_x e^{- H(x)/2} |x\rangle$.

\subsection{Boltzmann Machines}\label{hbms}
To present our results, we first introduce a formalism which we call a Hyper Boltzmann Machine (HBM). Let $G = (V,E)$ be a hypergraph with nodes $V$ and hyperedges $E$. A hyperedge is a set of nodes. We require each hyperedges $e \in E$ to contain at most $|e| \leq k$ nodes, and we will then describe the HBM as $k$-local. Each node is labelled either visible or hidden, so that $V = V_{\text{visible}} \cup V_{\text{hidden}}$. Say $|V_{\text{visible}}| = n$ and $|V_{\text{hidden}}| = m$. Each node (visible and hidden) carries a classical binary degree of freedom ie a bit. As a convention we label the visible node variables $x \in \{0,1\}^n$ and the hidden node variables $h \in \{0,1\}^m$. For each edge $e \in E$, we have a local energy term $F_e : e \rightarrow \mathbb{R}$. The total energy of the HBM is $F = \sum_e F_e$. We say the output of the HBM is the function:
$$ f(x) = \sum_h \exp(-F(x,h)) $$
where we sum over all possible values of the hidden variables $h \in \{0,1\}^m$. The HBM will act as an important intermediate between stoquastic ground states and the various classical thermal distributions.

Let $|\psi \rangle$ be a state on $n$ qubits. $x \in \{0,1\}^n$ now indexes the standard basis. We say a HBM represents the state $|\psi \rangle$ to precision $\epsilon$ in distribution if $\left| \left| \frac{f(x)}{\sum_y f(y)} - |\langle x|\psi \rangle |^2 \right| \right|_{\text{TV}} \leq \epsilon$, where $|| \cdot ||_{\text{TV}}$ is the total variation distance. We say a HBM represents the state $|\psi \rangle$ to precision $\epsilon$ in wavefunction if $\left| \left| \frac{f(x)}{\sqrt{\sum_y f(y)^2}} - \langle x|\psi \rangle  \right| \right|_2 \leq \epsilon$, where $|| \cdot ||_2$ is the 2-norm induced by the Hilbert space inner product. Note that for a state $|\psi \rangle$ to be represented by a HBM in wavefunction, it must have real positive amplitudes in the computational basis.

A Boltzmann Machine is a special case of the HBM, where the only allowed local energy terms are $F_{\{y\}} = - a_y y$ and $F_{\{y_1,y_2\}} = - W_{y_1,y_2} y_1 y_2$. In particular, we only allow hyperedges of size at most 2, so the hypergraph becomes a regular graph. A Deep Boltzmann Machine (DBM) is a further restriction, where we require the graph to be split into three layers: visible, middle and deep. The visible layer consists precisely of all the visible nodes, and we only allow edges between consecutive layers. Note that we do not gain any generalisation by allowing further deep layers: a $N$-layer DBM can be transformed into a 3-layer DBM by folding the even-numbered hidden layers into the middle layer, and the odd-numbered hidden layers into the deep layer. If the network has no nodes in the deep layer, we call it a Restricted Boltzmann Machine (RBM). Figure~\ref{boltzmannmachines} illustrates these concepts.

\begin{figure}[h]
     \centering
     \begin{subfigure}[H]{0.2\textwidth}
         \centering
         \includegraphics[width=\textwidth]{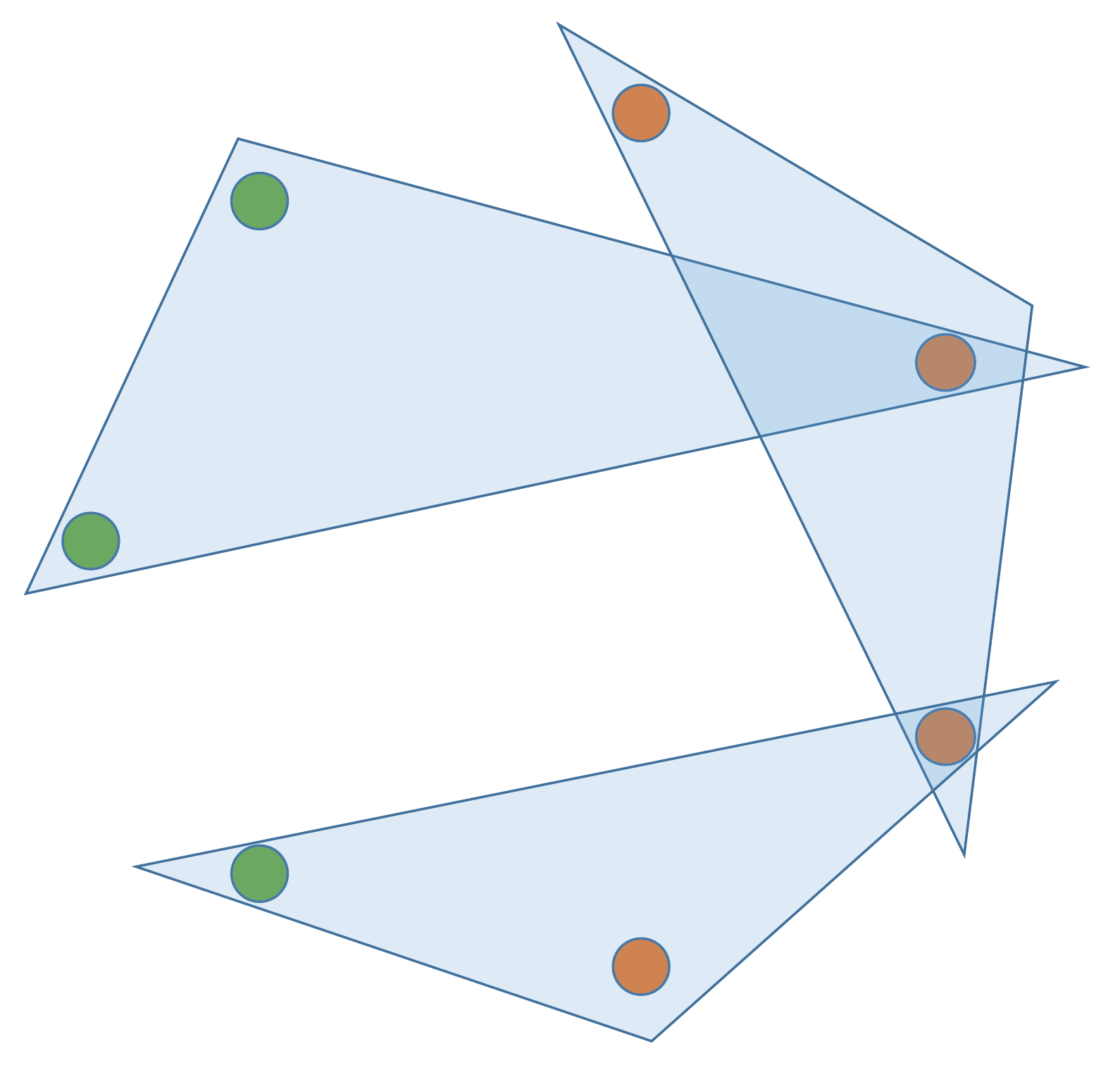}
         \caption{HBM}
     \end{subfigure}
     \begin{subfigure}[H]{0.2\textwidth}
         \centering
         \includegraphics[width=\textwidth]{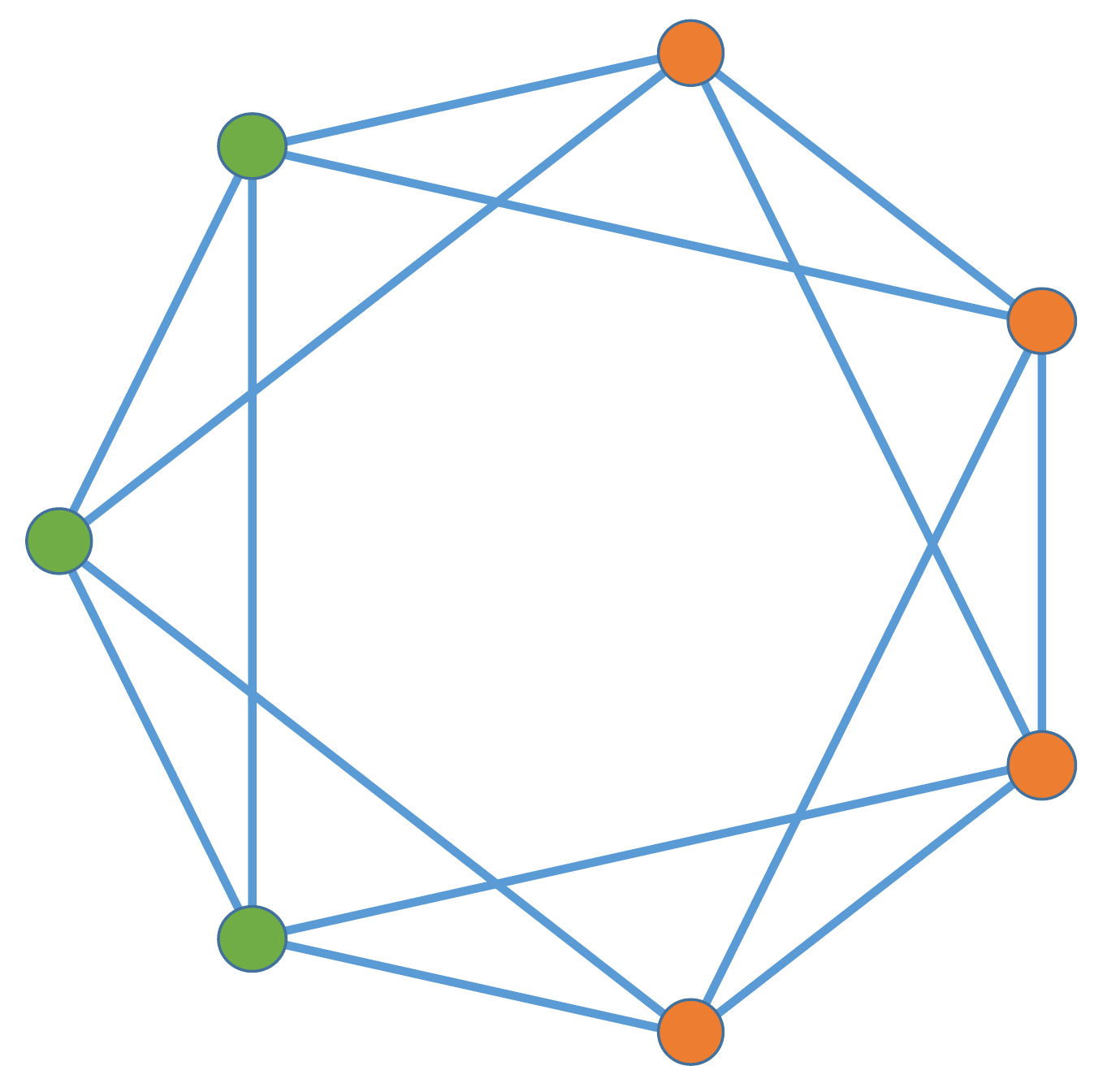}
         \caption{Boltzmann machine}
     \end{subfigure}
     
     \begin{subfigure}[H]{0.2\textwidth}
         \centering
         \includegraphics[width=\textwidth]{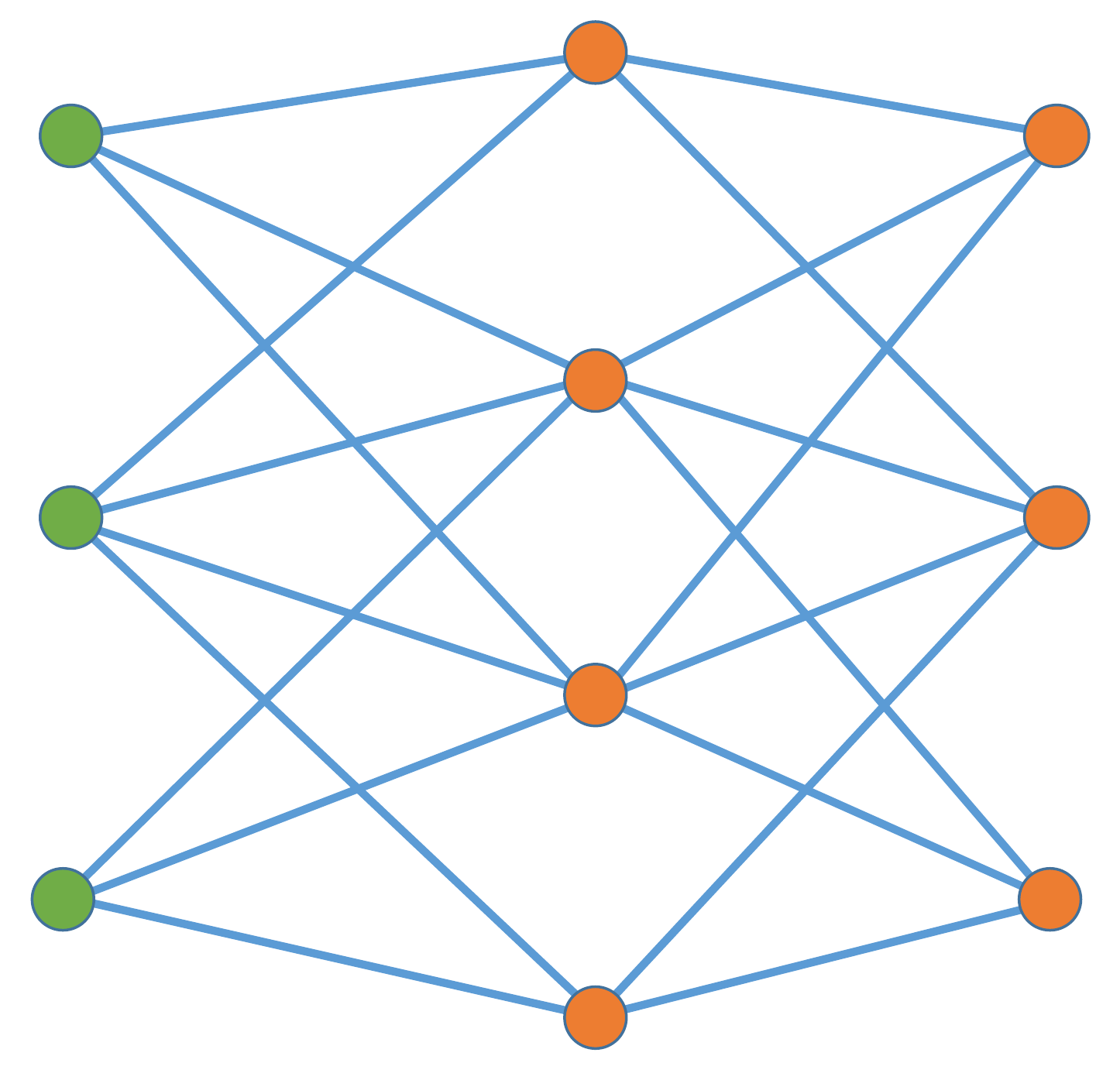}
         \caption{DBM}
     \end{subfigure}
     \begin{subfigure}[H]{0.15\textwidth}
         \centering
         \includegraphics[width=\textwidth]{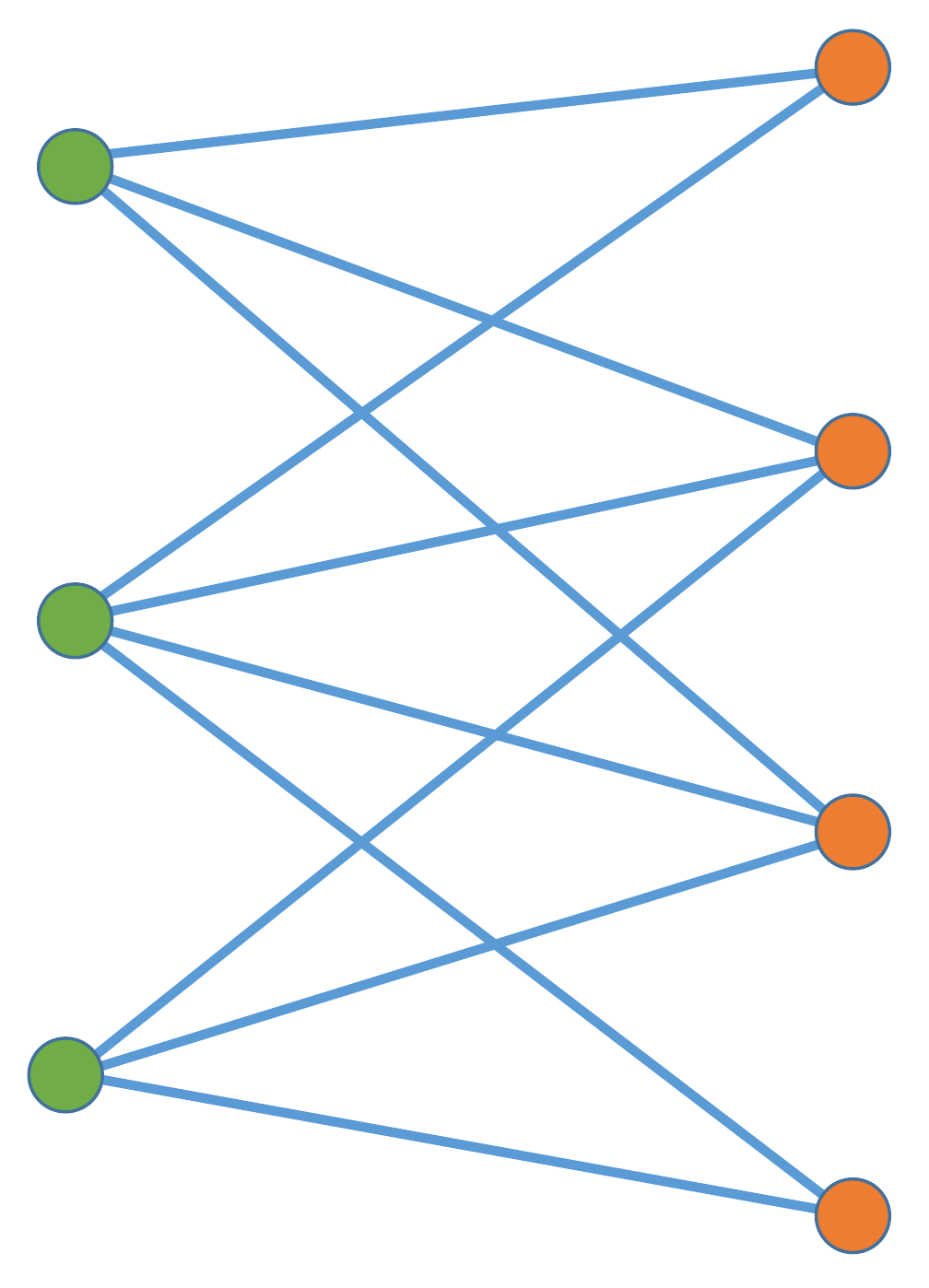}
         \caption{RBM}
     \end{subfigure}
     \caption{Various types of Boltzmann machine. The green nodes are the visible nodes, and the red nodes are the hidden nodes. The shaded patches in (a) are hyperedges, and the lines in (b),(c),(d) are edges.}
          \label{boltzmannmachines}
\end{figure}

We can make some initial observations:

\begin{observation}\label{observation}
\begin{enumerate}
	\item Consider a $k$-local HBM with $n$ visible nodes, $m$ hidden nodes, $T$ hyperedges, with each node contained in at most $k'$ hyperedges. The distribution given by this HBM is a marginal distribution of the Gibbs distribution of a $k$-local classical Hamiltonian on $n+m$ bits, with $T$ terms, where each bit is acted on by at most $k'$ terms.
	\item Consider the Gibbs distribution of a $k$-local classical Hamiltonian on $n$ bits, with $T$ terms, where each bit is acted on by at most $k'$ terms. This state is represented exactly in distribution by a $k$-local HBM with no hidden nodes, $T$ hyperedges, with each node contained in at most $k'$ hyperedges.
\end{enumerate}
\end{observation}
\begin{proof}
Identify the nodes of the HBM with the binary degrees of freedom (the bits) of the classical system. Then for each hyperedge $e$ with energy $F_e$, create a local term with the same energy on the corresponding bits for the classical Hamiltonian, and vice versa. Marginalising over the bits corresponding to hidden nodes then completes the equivalence.
\end{proof}

Throughout the paper, $|+\rangle$ will refer to the uniform superposition state, usually denoted $|+\rangle ^{\otimes n}$. The intuition behind the use of the $|+\rangle$ state is the following: (a) stoquastic Hamiltonian ground states have without loss of generality real positive amplitudes, thus they all have an overlap with $|+\rangle$; (b) it corresponds to the uniform distribution, which is represented by the empty HBM, and is the thermal distribution of a trivial classical Hamiltonian.

\bigskip

\section{Summary of main results}

We present our results in a series of theorems. First, we consider the task of simulating ground states of general local stoquastic Hamiltonians. Theorem \ref{stoq_Gibbs} shows the connection between their ground states and classical thermal distributions. Given a local stoquastic Hamiltonian $H$ with $1/\poly(n)$ energy gap and $\poly(n)$ operator norm (where $n$ denotes the number of qubits), we can find a local classical Hamiltonian $H_c$ such that the distribution of the ground state of $H$ is a known marginal distribution of the Gibbs distribution of $H_c$. For our proof to work, we introduce the technical condition that the overlap of the ground state and the $|+\rangle$ state is also polynomial in $n$.

We can compare this condition to that of Bravyi's `guiding state'~\cite{bravyi2014monte}. In one sense, it is a stronger condition since the state $|+\rangle$ is fixed -- we have no freedom to change it depending on the specific Hamiltonian in question. On the other hand, requiring a polynomial overlap is much weaker than the guiding state condition, which requires a componentwise polynomial relationship.

\begin{theorem}\label{stoq_Gibbs}
Let $H$ be a $k$-local stoquastic Hamiltonian on $n$ qubits, with $L = \poly(n)$ terms, energy gap $\Delta = 1/\poly(n)$, with each term having operator norm bounded by $J = \poly(n)$. Let $\Pi$ project onto the ground subspace of $H$, and assume $\eta = \sqrt{\langle +|\Pi |+ \rangle} \geq 1/\poly(n)$. The distribution of the ground state $|\psi_0 \rangle = \frac{1}{\eta} \Pi |+\rangle$ can be approximated to precision $\epsilon$ (in total variation) by a marginal distribution of the Gibbs distribution of a $2k$-local classical Hamiltonian on $O(k \eta^{-1/2} L^{5/2} J^{3/2} \Delta^{-3/2} \epsilon^{-1/2} (\log(\eta^{-1} \epsilon^{-1}))^{3/2})$ bits, with $O(\eta^{-1/2} L^{5/2} J^{3/2} \Delta^{-3/2} \epsilon^{-1/2} (\log(\eta^{-1} \epsilon^{-1}))^{3/2})$ terms, and where each bit is acted on by at most 2 terms.
\end{theorem}

When we additionally impose the condition that the local stoquastic Hamiltonian is frustration free, then we can prove a stronger result. Here we do not require any conditions on the overlap of the ground subspace with $|+\rangle$.

\begin{theorem}\label{SFF_Gibbs}
Let $H$ be a $k$-local SFF Hamiltonian on $n$ qubits, with $L = \poly(n)$ terms, energy gap $\Delta = 1/\poly(n)$, with each term having operator norm bounded by $J = \poly(n)$. Let $\Pi$ project onto the ground subspace of $H$. The distribution of the ground state $|\psi_0 \rangle = \frac{1}{\eta} \Pi |+ \rangle$ can be approximated to precision $\epsilon$ (in total variation) by the marginal distribution of the Gibbs distribution of a $2k$-local classical Hamiltonian on at most $O(k L^3 J \Delta^{-1} \log(\eta^{-1} \epsilon^{-1}))$ bits, with $O(L^3 J \Delta^{-1} \log(\eta^{-1} \epsilon^{-1}))$ terms, and where each bit is acted on by at most 2 terms. Note $\log(\eta^{-1}) = \poly(n)$ (see Lemma \ref{SFF_P}).
\end{theorem}

We then turn to the converse: given a local classical Hamiltonian $H_c$, can we find a SFF Hamiltonian $H$ which has as its unique ground state a coherent version of the Gibbs distribution of $H_c$? This question has been answered affirmatively by the results of Verstraete et al. \cite{verstraete2006criticality} and Somma et al. \cite{somma2007quantum}, which we summarise in Theorem \ref{Gibbs_SFF}.

\begin{theorem}\label{Gibbs_SFF}
Consider the $n$-qubit coherent Gibbs state on a $k$-local classical Hamiltonian, with $T$ terms, where each qubit is acted on by at most $k'$ terms. This Gibbs state is the unique ground state of a $(k'(k-1)+1)$-local SFF Hamiltonian with $n$ terms.
\end{theorem}

Theorem \ref{Gibbs_SFF} has two immediate corollaries. From Theorem \ref{stoq_Gibbs} we know that a stoquastic ground state is simulated by a classical thermal distribution, and Theorem \ref{Gibbs_SFF} says that coherent states corresponding to classical thermal distributions are SFF ground states. Combining these two results has the following implication for general local stoquastic Hamiltonian: given a local stoquastic Hamiltonian $H$, we can find a SFF Hamiltonian $H'$ on a larger set of qubits such that the marginal distribution of the ground state of $H'$ approximates the distribution of the ground state of $H$. If $H$ is $k$-local, the classical Hamiltonian from Theorem \ref{stoq_Gibbs} is $2k$-local, so $H'$ from Theorem \ref{Gibbs_SFF} is $4k$-local (we have $k'=2$ in this case). To our knowledge, this is the first efficient embedding of the ground state of an arbitrary local stoquastic Hamiltonian into the local SFF Hamiltonian.

\begin{corollary}\label{stoq_SFF}
Let $H$ be a $k$-local stoquastic Hamiltonian on $n$ qubits, with $L = \poly(n)$ terms, energy gap $\Delta = 1/\poly(n)$, with each term having operator norm bounded by $J = \poly(n)$. Let $\Pi$ project onto the ground subspace of $H$, and assume $\eta = \sqrt{\langle +|\Pi |+ \rangle} \geq 1/\poly(n)$. The distribution of the ground state $|\psi_0 \rangle = \frac{1}{\eta} \Pi |+ \rangle$ can be approximated to precision $\epsilon$ (in total variation) by a marginal distribution of the unique ground state of a $4k$-local SFF Hamiltonian on $O(k \eta^{-1/2} L^{5/2} J^{3/2} \Delta^{-3/2} \epsilon^{-1/2} (\log(\eta^{-1} \epsilon^{-1}))^{3/2})$ qubits, with a term for each qubit.
\end{corollary}

From Observation~\ref{observation} we know that HBMs are equivalent to classical thermal distributions, and in Section~\ref{hbms} we showed that all Boltzmann machines are HBMs. This together with the result of Theorem \ref{Gibbs_SFF} leads to the following Corollary \ref{BM_SFF}: given a Boltzmann machine representing distribution $f$, we can find a SFF Hamiltonian $H$ such that $f$ is a marginal distribution of the ground state of $H$. The classical Hamiltonian corresponding to a Boltzmann machine is 2-local, thus $H$ is $(k'+1)$-local.

\begin{corollary}\label{BM_SFF}
Consider a Boltzmann machine with $n$ visible nodes and $m$ hidden nodes, where each node has at most $k'$ connections. The distribution given by this Boltzmann machine is a marginal distribution of the unique ground state of a $(k'+1)$-local SFF Hamiltonian on $n+m$ qubits, with $n+m$ terms.
\end{corollary}

The following results investigate the representability of the ground state distributions using a particular class of Boltzmann machines. Theorem \ref{stoq_DBM} and \ref{SFF_DBM} show how to represent the distribution of the ground state of a stoquastic and SFF Hamiltonian  using a DBM.

\begin{theorem}\label{stoq_DBM}
Let $H$ be a $k$-local stoquastic Hamiltonian on $n$ qubits, with $L = \poly(n)$ terms, energy gap $\Delta = 1/\poly(n)$, with each term having operator norm bounded by $J = \poly(n)$. Let $\Pi$ project onto the ground subspace of $H$, and assume $\eta = \sqrt{\langle +|\Pi |+ \rangle} \geq 1/\poly(n)$. The distribution of the ground state $|\psi_0 \rangle = \frac{1}{\eta} \Pi |+\rangle$ can be represented to precision $\epsilon$ (in total variation) by a DBM with $O(2^{2k} \eta^{-1/2} L^{5/2} J^{3/2} \Delta^{-3/2} \epsilon^{-1/2} (\log(\eta^{-1} \epsilon^{-1}))^{3/2})$ hidden nodes, where each node has at most $2^{2k+1}$ connections.
\end{theorem}

Similarly to the above, when we restrict our stoquastic Hamiltonian to be frustration free, this representation becomes simpler.

\begin{theorem}\label{SFF_DBM}
Let $H$ be a $k$-local SFF Hamiltonian on $n$ qubits, with $L = \poly(n)$ terms, energy gap $\Delta = 1/\poly(n)$, with each term having operator norm bounded by $J = \poly(n)$. Let $\Pi$ project onto the ground subspace of $H$. The distribution of the ground state $|\psi_0 \rangle = \frac{1}{\eta} \Pi |+ \rangle$ can be represented to precision $\epsilon$ (in total variation) by a DBM with $O(2^{2k} L^3 J \Delta^{-1} \log(\eta^{-1} \epsilon^{-1}))$ hidden nodes, where each node has at most $2^{2k+1}$ connections. Note $\log(\eta^{-1}) = \poly(n)$ (see Lemma \ref{SFF_P}).
\end{theorem}

Lastly, we show that for classical thermal distribution of local Hamiltonians it suffices to restrict the model to a RBM.

\begin{theorem}\label{Gibbs_RBM}
Consider the Gibbs distribution of a $k$-local classical Hamiltonian on $n$ bits, with $T$ terms, where each bit is acted on by at most $k'$ terms. This distribution can be represented to arbitrary precision* (in total variation) by a RBM with $n$ visible nodes, at most $2^k T$ hidden nodes, where each node has at most $k' 2^k$ connections.
\end{theorem}

It should be noted that all of the above mappings are explicit: in the proofs, we explicitly construct the mappings.

{\bf Remark*}: by stating that a system $B$ can approximate another system $A$ to `arbitrary precision' we mean that $B$ approximates $A$ with error $\epsilon$, and the size and connectivity of $B$ do not depend on $\epsilon$. The Hamiltonian of $B$, however, will depend on $\epsilon$ in general. In all cases, if we set $\epsilon = 1/\poly(n)$, it suffices to take a Hamiltonian whose values are bounded by $\poly(n)$.

\bigskip

\section{HBMs can represent stoquastic ground states}

In order to prove Theorems \ref{stoq_Gibbs} and \ref{SFF_Gibbs}, we would like to show that the ground states of stoquastic and SFF Hamiltonians respectively can be represented in distribution by HBMs. We can then recall Observation~\ref{observation} to deduce that stoquastic/SFF ground states can be approximated by a marginal distribution of the Gibbs distribution of a local classical Hamiltonian.

\begin{theorem}\label{stoq_HBM}
Let $H$ be a $k$-local stoquastic Hamiltonian on $n$ qubits, with $L = \poly(n)$ terms, energy gap $\Delta = 1/\poly(n)$, with each term having operator norm bounded by $J = \poly(n)$. Let $\Pi$ project onto the ground subspace of $H$, and assume $\eta = \sqrt{\langle +|\Pi |+ \rangle} \geq 1/\poly(n)$. The ground state $|\psi_0 \rangle = \frac{1}{\eta} \Pi |+ \rangle$ can be represented to precision $\epsilon$ in distribution by a $2k$-local HBM with $O(k \eta^{-1/2} L^{5/2} J^{3/2} \Delta^{-3/2} \epsilon^{-1/2} (\log(\eta^{-1} \epsilon^{-1}))^{3/2})$ hidden nodes, $O(\eta^{-1/2} L^{5/2} J^{3/2} \Delta^{-3/2} \epsilon^{-1/2} (\log(\eta^{-1} \epsilon^{-1}))^{3/2})$ hyperedges, where each node is contained in at most 2 hyperedges.
\end{theorem}

\begin{theorem}\label{SFF_HBM}
Let $H$ be a $k$-local SFF Hamiltonian on $n$ qubits, with $L = \poly(n)$ terms, energy gap $\Delta = 1/\poly(n)$, with each term having operator norm bounded by $J = \poly(n)$. Let $\Pi$ project onto the ground subspace of $H$. The ground state $|\psi_0 \rangle = \frac{1}{\eta} \Pi |+ \rangle$ can be represented to precision $\epsilon$ in distribution by a $2k$-local HBM with $O(k L^3 J \Delta^{-1} \log(\eta^{-1} \epsilon^{-1}))$ hidden nodes, $O(L^3 J \Delta^{-1} \log(\eta^{-1} \epsilon^{-1}))$ hyperedges, where each node is contained in at most 2 hyperedges. Note $\log(\eta^{-1}) = \poly(n)$ (see Lemma \ref{SFF_P}).
\end{theorem}

In the rest of this section, we will prove these theorems via a sequence of lemmas. Our strategy will be as follows: we will first construct a sequence of $k$-local entrywise positive matrices $P_1, \dots, P_T$ which act to converge the state $|+\rangle$ to a ground state of the given Hamiltonian. We do this in the stoquastic and SFF cases separately, in Lemmas \ref{stoq_P} and \ref{SFF_P} respectively. In the SFF case, the sequence will project onto the ground subspace of each term separately. In the stoquastic case, we must use a Trotter decomposition for the imaginary time evolution operator. We then use this sequence to find a HBM with output $f$ which represents the ground state in wavefunction, which is done in Lemma \ref{P_HBM}. The idea in Lemma \ref{P_HBM} is to represent the action of a $k$-local entrywise positive matrix $P$ by adding some new hidden nodes, and a new hyperedge. We then must convert this to a HBM with output $f^2$ which represents the ground state in distribution, which is achieved by `squaring' the HBM. This is done in Lemma \ref{squareHBM}.

\begin{lemma}\label{stoq_P}
Let $H$ be a $k$-local stoquastic Hamiltonian on $n$ qubits, with $L = \poly(n)$ terms, energy gap $\Delta = 1/\poly(n)$, with each term having operator norm bounded by $J = \poly(n)$. Let $\Pi$ project onto the ground subspace of $H$, and assume $\eta = \sqrt{\langle +| \Pi |+ \rangle} \geq 1/\poly(n)$. We can find matrices $P_1, \dots, P_T$ (not necessarily unitary) satisfying
\begin{itemize}
	\item $P_i$ is $k$-local.
	\item The entries of $P_i$ as a $k \times k$ matrix are real and positive.
\end{itemize}
such that the state:
$$ |\psi \rangle = \frac{1}{\sqrt{Z}} P_T \dots P_1 |+ \rangle $$
is within $\left| \left| |\psi \rangle - |\psi_0 \rangle \right| \right|_2 < \epsilon$ of the ground state $|\psi_0 \rangle = \frac{1}{\eta} \Pi |+\rangle$. ($Z$ is a normalization constant.) We have $T = O(\eta^{-1/2} L^{5/2} J^{3/2} \Delta^{-3/2} \epsilon^{-1/2} (\log(\eta^{-1} \epsilon^{-1}))^{3/2})$.
\end{lemma}
\begin{proof}
Assume the ground energy of $H$ is zero. This is without loss of generality, since we can add scalar multiples of the identity. The imaginary time evolution operator $e^{-\tau H}$ has operator norm 1. If we apply $e^{-\tau H}$ to $|+\rangle$, we have:
$$ \left| \left| \frac{1}{\eta} e^{-\tau H} |+\rangle - |\psi_0 \rangle \right| \right|_2 = O(\eta^{-1} e^{-\tau \Delta}) $$

To see this, write $|+\rangle = \eta |\psi_0 \rangle + \sqrt{1 - \eta^2} |\psi_0 \rangle^\perp$, where $|\psi_0 \rangle^\perp$ is some unit vector orthogonal to $|\psi_0 \rangle$. Then $\left| \left| \frac{1}{\eta} e^{-\tau H} |+\rangle - |\psi_0 \rangle \right| \right|_2 = \left| \left| \frac{\sqrt{1-\eta^2}}{\eta} e^{-\tau H} |\psi_0 \rangle^\perp \right| \right|_2 \leq \frac{\sqrt{1-\eta^2}}{\eta} e^{-\tau \Delta} \leq \eta^{-1} e^{-\tau \Delta}$

Consider the second order Suzuki-Trotter decomposition $e^{- \delta H} \approx e^{- \frac{1}{2} \delta H_L} \dots e^{- \frac{1}{2} \delta H_1} e^{- \frac{1}{2} \delta H_1} \dots e^{- \frac{1}{2} \delta H_L}$ for small $\delta = 1/\poly(n)$. 
The error is:
\begin{align*}
\left| \left| e^{- \frac{1}{2} \delta H_L} \dots e^{- \frac{1}{2} \delta H_1} e^{- \frac{1}{2} \delta H_1} \dots e^{- \frac{1}{2} \delta H_L} - e^{-\delta H} \right| \right|_{\text{op}} \\ = O(\delta^3 L^3 J^3) 
\end{align*}
For a proof of this, see Appendix \ref{Trotter}.

Let $h_j$ be the $k\times k$ matrix consisting of a $1$ in each entry, acting on the same qubits as $H_j$. Let $\alpha$ be small.

\begin{align*} 
&\left( e^{- \frac{1}{2} \delta H_L} + \alpha h_L \right) \dots \left( e^{- \frac{1}{2} \delta H_1} + \alpha h_1 \right) \times \\
&\times \left( e^{- \frac{1}{2} \delta H_1} + \alpha h_1 \right) \dots \left( e^{- \frac{1}{2} \delta H_L} + \alpha h_L  \right) \\
&= e^{- \frac{1}{2} \delta H_L} \dots e^{- \frac{1}{2} \delta H_1} e^{- \frac{1}{2} \delta H_1} \dots e^{- \frac{1}{2} \delta H_L} + O(\alpha L)\\
&= e^{-\delta H} + O(\delta^3 L^3 J^3) + O(\alpha L)
\end{align*}

\begin{align*} 
& \Big[ \left( e^{- \frac{1}{2} \delta H_L} + \alpha h_L \right) \dots \left( e^{- \frac{1}{2} \delta H_1} + \alpha h_1 \right) \times \\ 
&\times \left( e^{- \frac{1}{2} \delta H_1} + \alpha h_1 \right) \dots \left( e^{- \frac{1}{2} \delta H_L} + \alpha h_L \right) \Big]^N \\ 
&= e^{-N \delta H} + O(N \delta^3 L^3 J^3) + O(N \alpha L) 
\end{align*}

Let 
\begin{align*}
P_T &\dots P_1 = \\ & \Big[ \left( e^{- \frac{1}{2} \delta H_L} + \alpha h_L \right) \dots \left( e^{- \frac{1}{2} \delta H_1} + \alpha h_1 \right) \times \\ &\times \left( e^{- \frac{1}{2} \delta H_1} + \alpha h_1 \right) \dots \left( e^{- \frac{1}{2} \delta H_L} + \alpha h_L \right) \Big]^N.
\end{align*}

 So $P_i = e^{- \frac{1}{2} \delta H_j} + \alpha h_j$ for some $j$. Note the entries of $e^{- \frac{1}{2} \delta H_j}$ are non-negative by stoquasticity of $H_j$, so the entries of $P_i$ are positive (bigger than $\alpha$). Note also $T = 2NL$. We have:
$$ \left| \left| P_T \dots P_1 - e^{- N \delta H} \right| \right|_{\text{op}} = O(N \delta^3 L^3 J^3) + O(N \alpha L) $$

Thus we have:
\begin{align*}
&\left| \left| \frac{1}{\eta} P_T \dots P_1 |+\rangle - |\psi_0 \rangle \right| \right|_2 \leq \\ & \left| \left| \frac{1}{\eta} P_T \dots P_1 |+\rangle - \frac{1}{\eta} e^{-N \delta H} |+\rangle \right| \right|_2 + \left| \left| \frac{1}{\eta} e^{-N \delta H} |+\rangle - |\psi_0 \rangle \right| \right|_2 \\ 
&= O(\eta^{-1} N \delta^3 L^3 J^3) + O(\eta^{-1} N \alpha L) + O(\eta^{-1} e^{-N \delta \Delta})
\end{align*}

By Lemma \ref{norms}:
\begin{align*}
&\left| \left| \frac{1}{\sqrt{Z}} P_T \dots P_1 |+\rangle - |\psi_0 \rangle \right| \right|_2 = \\ &O(\eta^{-1} N \delta^3 L^3 J^3) + O(\eta^{-1} N \alpha L) + O(\eta^{-1} e^{-N \delta \Delta}) 
\end{align*}

Thus if we want the total error to be bounded by $\epsilon$, we must take:

$$ \delta = O\big(\eta^{1/2} L^{-3/2} J^{-3/2} \Delta^{1/2} \epsilon^{1/2} (\log (\eta^{-1} \epsilon^{-1}))^{-1/2}\big) $$
$$ N = O\big(\eta^{-1/2} L^{3/2} J^{3/2} \Delta^{-3/2} \epsilon^{-1/2} (\log (\eta^{-1} \epsilon^{-1}))^{3/2}\big) $$
$$ \alpha = O\big( \eta N^{-1} L^{-1} \epsilon \big) = O\big( \eta T^{-1} \epsilon \big)$$

The number of local terms is:

$$ T = 2NL = O\big( \eta^{-1/2} L^{5/2} J^{3/2} \Delta^{-3/2} \epsilon^{-1/2} (\log(\eta^{-1} \epsilon^{-1}))^{3/2}\big) $$
\end{proof}

\begin{lemma}\label{SFF_P}
Let $H$ be a $k$-local SFF Hamiltonian on $n$ qubits, with $L = \poly(n)$ terms, energy gap $\Delta = 1/\poly(n)$, with each term having operator norm bounded by $J = \poly(n)$. Let $\Pi$ project onto the ground subspace of $H$. We can find matrices $P_1, \dots, P_T$ (not necessarily unitary) satisfying
\begin{itemize}
	\item $P_i$ is $k$-local.
	\item The entries of $P_i$ as a $k \times k$ matrix are real and positive.
\end{itemize}
such that the state:
$$ |\psi \rangle = \frac{1}{\sqrt{Z}} P_T \dots P_1 |+ \rangle $$
is within $|| |\psi \rangle - |\psi_0 \rangle ||_2 < \epsilon$ of the ground state $|\psi_0 \rangle = \frac{1}{\eta} \Pi |+\rangle$. ($Z$ is a normalization constant.) We have $T=O(L J \Delta^{-1} \log(\eta^{-1} \epsilon^{-1}))$. Note $\log(\eta^{-1}) = \poly(n)$ (see proof).
\end{lemma}
\begin{proof}
First recall that a ground state of a stoquastic Hamiltonian has without loss of generality real positive amplitudes in the computational basis. This implies the overlap $\eta = \sqrt{\langle +| \Pi |+ \rangle} \geq 2^{-n/2}$. Let $\Pi_j$ project onto the ground subspace of $H_j$. Applying the Detectability Lemma from~\cite{anshu2016simple,aharonov2009detectability}, we have that
$$ \left| \left| \Pi - \left[ \Pi_L \dots \Pi_1 \right]^N \right| \right|_{\text{op}} = O(e^{-\frac{1}{4} N \Delta J^{-1} L^{-2}}) $$

Let $h_j$ be the $k\times k$ matrix consisting of a $1$ in each entry, acting on the same qubits as $H_j$. Let $\alpha$ be small.
\begin{align*}
&\left| \left| \Pi - \left[ (\Pi_L + \alpha h_L) \dots (\Pi_1 + \alpha h_1) \right]^N \right| \right|_{\text{op}} \\ &=  O(N \alpha L) + O(e^{-\frac{1}{4} N \Delta J^{-1} L^{-2}})
\end{align*}

Let
$$ P_T \dots P_1 = \left[ (\Pi_L + \alpha h_L) \dots (\Pi_1 + \alpha h_1) \right]^N $$

So $P_i = \Pi_j + \alpha h_j$ for some $j$. Note the entries of $\Pi_j$ are non-negative by stoquasticity of $H_j$, and so the entries of $P_j$ are positive (bigger than $\alpha$). Note also $T=NL$. We have:
$$ \left| \left| P_T \dots P_1 - \Pi \right| \right|_{\text{op}} = O(N \alpha L) + O(e^{-\frac{1}{4} N \Delta J^{-1} L^{-2}}) $$

Thus:
\begin{align*}
&\left| \left| \frac{1}{\eta} P_T \dots P_1 |+\rangle - |\psi_0 \rangle \right| \right|_2 \\&= \left| \left| \frac{1}{\eta} (P_L \dots P_1 - \Pi) |+\rangle \right| \right|_2 \\&= O(\eta^{-1} N \alpha L) + O(\eta^{-1} e^{-\frac{1}{4} N \Delta J^{-1} L^{-2}})
\end{align*}

By Lemma \ref{norms}:
\begin{align*}
&\left| \left| \frac{1}{\sqrt{Z}} P_L \dots P_1 |+\rangle - |\psi_0 \rangle \right| \right|_2 \\ &= O(\eta^{-1} N \alpha L) + O(\eta^{-1} e^{-\frac{1}{4} N \Delta J^{-1} L^{-2}})
\end{align*}

To have the total error be bounded by $\epsilon$, we take:
$$ N = O(L^2 J \Delta^{-1} \log(\eta^{-1} \epsilon^{-1})) $$
$$ \alpha = O(\eta N^{-1} L^{-1} \epsilon) = O(\eta T^{-1} \epsilon) $$

The number of local terms is:
$$ T = NL = O(L^3 J \Delta^{-1} \log(\eta^{-1} \epsilon^{-1})) $$
\end{proof}

\begin{lemma}\label{P_HBM}
Let $P_1, \dots, P_T$ be matrices (not necessarily unitary) satisfying
\begin{itemize}
	\item $P_i$ is $k$-local.
	\item The entries of $P_i$ as a $k \times k$ matrix are real and positive.
\end{itemize}
The state
$$ |\psi \rangle = \frac{1}{\sqrt{Z}} P_T \dots P_1 |+ \rangle $$
can be represented exactly in wavefunction by a $2k$-local HBM with $kT$ hidden nodes, $T$ hyperedges, each node is contained in at most 2 hyperedges.
\end{lemma}
\begin{proof}
Suppose we have a HBM representing exactly in wavefunction the state $|\psi \rangle$. We are given a $k$-local matrix $P$ with real positive entries in the computational basis. We can find a new HBM which represents exactly in wavefunction the state $|\psi'\rangle = \frac{1}{\sqrt{Z}} P |\psi \rangle$.

To see this, we will update the existing HBM in a way which represents the action of $P$. Write $x = (\tilde{x}, x')$ with $|x'| = k$, $|\tilde{x}| = n-k$, and say the matrix $P$ acts on qubits $x'$. We perform the following:
\begin{enumerate}
	\item Turn the $k$ visible nodes $x'$ into new hidden nodes $h'$, ie keeping the same energy terms.
	\item Replace the $k$ visible nodes $x'$.
	\item Add a hyperedge $e'$ on the $2k$ vertices $\{ x', h' \}$, with local energy:
		$$ F_{e'} (x', h') = - \log{ \langle x' |P| h' \rangle } $$
\end{enumerate}

The original HBM has output
$$ f(\tilde{x}, x') = \sum_h \exp(-F(\tilde{x}, x', h)) \propto \langle \tilde{x}, x' |\psi \rangle $$

The new HBM has output
\begin{align*} 
f'(\tilde{x}, x') &= \sum_{h, h'} \exp(-F_{e'} (x', h')) \exp(-F(\tilde{x}, h', h))\\&\propto \sum_{h'} \langle x' |P| h' \rangle \langle \tilde{x}, h' | \psi \rangle \\&\propto \langle \tilde{x}, x' | \psi' \rangle 
\end{align*}
as desired.

The empty HBM represents exactly the state $|+\rangle$ in wavefunction. Thus if we apply the above construction iteratively on $P_1, \dots, P_T$, starting with an empty HBM, we get the desired HBM. By examining the HBM updates, we can see that the resulting HBM will be $2k$-local with $kT$ hidden nodes, $T$ hyperedges, and each node contained in at most 2 hyperedges.
\end{proof}

\begin{figure}[h]
\centering
  \includegraphics[width=0.5\textwidth]{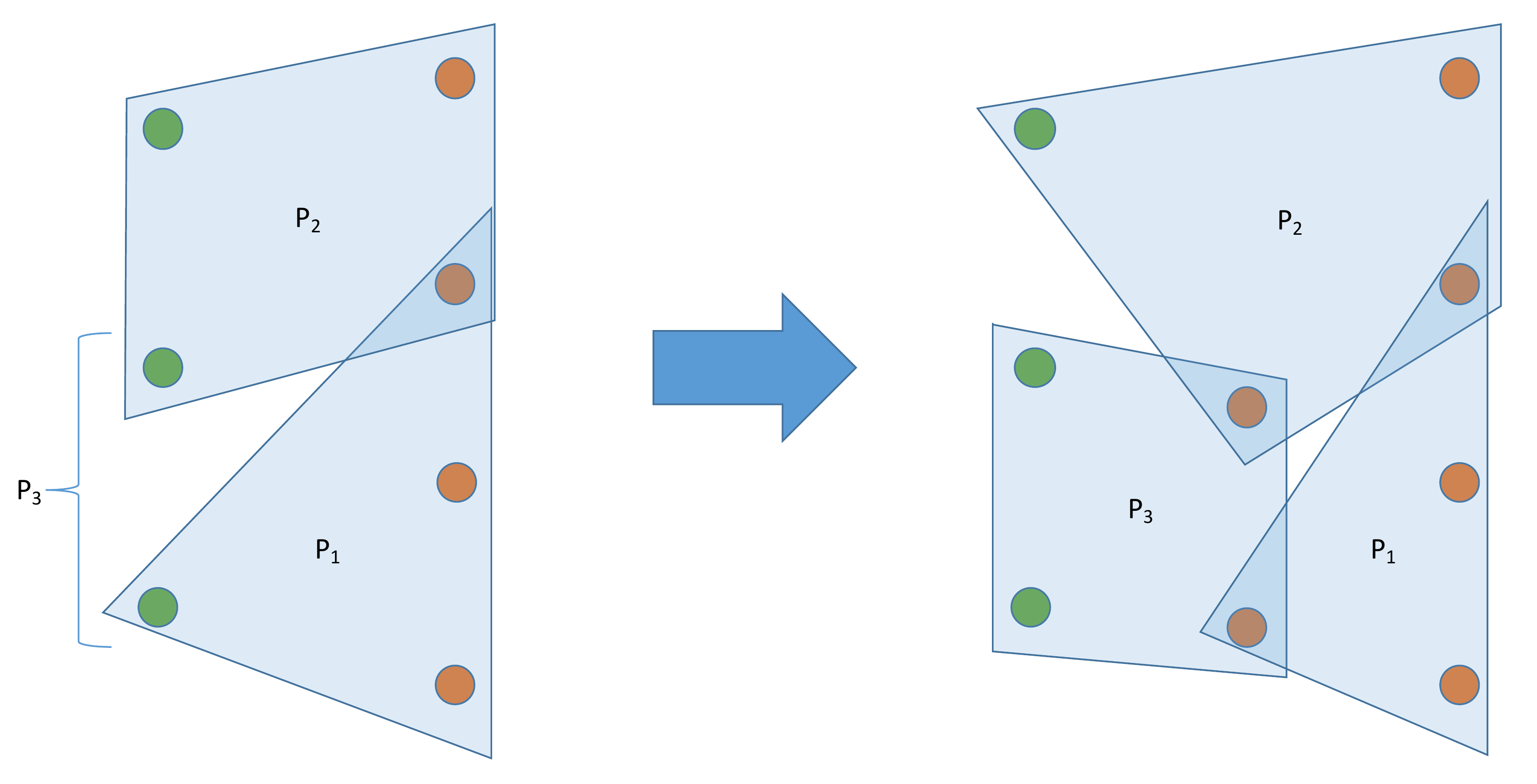}
  \caption{Consider a 3 qubit system. Say $P_1$ acts on qubits 2,3, $P_2$ acts on qubits 1,2 and $P_3$ acts again on qubits 2,3. The left HBM represents the state $\frac{1}{\sqrt{Z}} P_2 P_1 |+\rangle$, and the right HBM represents the state $\frac{1}{\sqrt{Z'}} P_3 P_2 P_1 |+\rangle$. We have added replicas of qubits 2,3 and an extra hyperedge corresponding to the operator $P_3$.}
\
\end{figure}

{\bf Remark*}: The energy terms $F_e$ in applications of Lemma \ref{P_HBM} to Lemma \ref{stoq_P} and \ref{SFF_P} are bounded by $\log (\alpha^{-1}) = \poly(n)$.

\begin{lemma}\label{squareHBM}
Given a HBM with output $f(x)$ which represents a state $|\psi \rangle$ to precision $\epsilon$ in wavefunction ie 
$$ \left| \left| \frac{f(x)}{\sqrt{\sum_y f(y)^2}} - \langle x|\psi \rangle \right| \right|_2 \leq \epsilon $$
we can find a HBM with output $f'(x)=f(x)^2$ which represents $|\psi \rangle$ to precision $\epsilon$ in distribution ie 
$$ \left| \left| \frac{f'(x)}{\sum_y f'(y)} - | \langle x|\psi \rangle |^2 \right| \right|_{\text{TV}} \leq \epsilon $$
\end{lemma}
\begin{proof}
Let $f(x)$ be the output of the original HBM. We wish to find a HBM with output $f'(x) = f(x)^2$. We do this by `squaring' the original HBM: We duplicate the hidden nodes $h \rightarrow (h, h')$, copying also the local energy terms, and connect them to the same visible nodes.
$$f(x) = \sum_h e^{-F(x,h)}$$
\begin{align*}
f'(x) &= \sum_{h, h'} e^{-F(x,h) - F(x,h')} \\&= \left( \sum_h e^{-F(x,h)} \right) \left( \sum_{h'} e^{-F(x,h')} \right) = f(x)^2 
\end{align*}

If we examine the errors, we have by assumption
$$ \left| \left| \frac{f(x)}{\sqrt{\sum_y f(y)^2}} - \langle x|\psi \rangle \right| \right|_2 \leq \epsilon $$

This implies
\begin{align*}
&\left| \left| \frac{f'(x)}{\sum_y f'(y)} - |\langle x|\psi \rangle |^2  \right| \right|_{\text{TV}} = \frac{1}{2} \sum_x \left| \frac{f(x)^2}{\sum_y f(y)^2} - |\langle x|\psi \rangle |^2 \right| \\
&= \frac{1}{2} \sum_x \left( \left| \frac{f(x)}{\sqrt{\sum_y f(y)^2}} - \langle x|\psi \rangle \right| \left( \frac{f(x)}{\sqrt{\sum_y f(y)^2}} + \langle x|\psi \rangle \right) \right) \\
&\leq \frac{1}{2} \left| \left| \frac{f(x)}{\sqrt{\sum_y f(y)^2}} - \langle x|\psi \rangle \right| \right|_2 \left| \left| \frac{f(x)}{\sqrt{\sum_y f(y)^2}} + \langle x|\psi \rangle \right| \right|_2 \leq \epsilon
\end{align*}
by the Cauchy-Schwarz inequality.
\end{proof}

\begin{figure}[h]
\centering
  \includegraphics[width=0.5\textwidth]{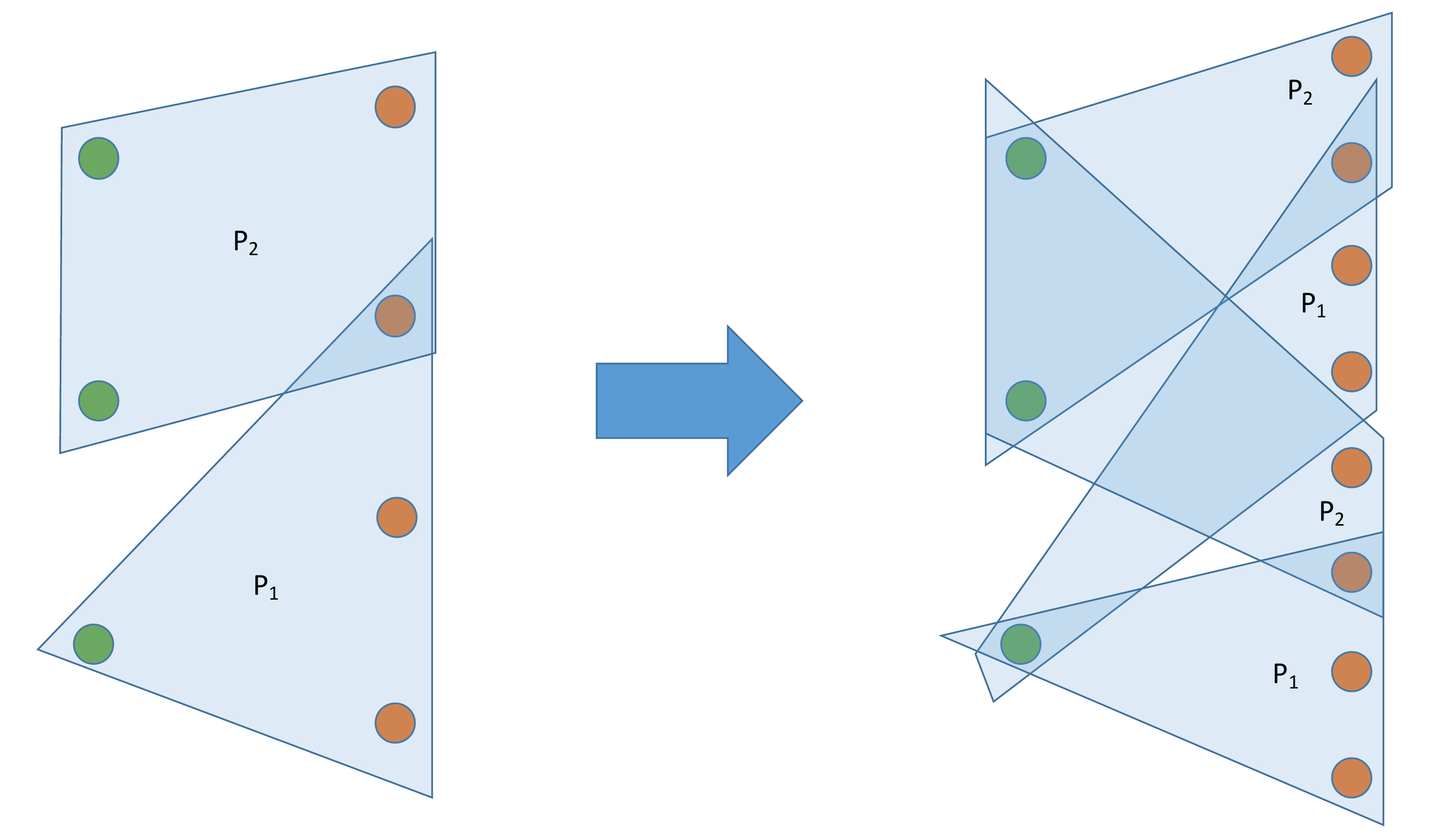}
  \caption{In a 3 qubit system $P_1$ acts on qubits 2,3 and $P_2$ acts on qubits 1,2. The HBM on the left represents the state $\frac{1}{\sqrt{Z}} P_2 P_1 |+\rangle$ in wavefunction, and the HBM on the right represents the state $\frac{1}{\sqrt{Z}} P_2 P_1 |+\rangle$ in distribution. We have squared the HBM by duplicating each hidden node, along with each hyperedge.}
\
\end{figure}

Note that if we square the HBM obtained from Lemma \ref{P_HBM}, it will remain $2k$-local with each node contained in at most 2 hyperedges, and we double the number of hidden nodes and hyperedges.

As discussed above, Lemmas \ref{stoq_P}, \ref{P_HBM}, \ref{squareHBM} prove Theorem \ref{stoq_HBM}; Lemmas \ref{SFF_P}, \ref{P_HBM}, \ref{squareHBM} prove Theorem \ref{SFF_HBM}; Theorem \ref{stoq_HBM} and Observation \ref{observation} prove Theorem \ref{stoq_Gibbs}; and Theorem \ref{SFF_HBM} and Observation \ref{observation} prove Theorem \ref{SFF_Gibbs}.

\bigskip

\section{All Gibbs states are SFF ground states}\label{Bravyi_Terhal_SFF}

In this section we prove Theorem \ref{Gibbs_SFF}. We follow the approach in Bravyi and Terhal \cite{bravyi2010complexity}, which is in turn based on the results of Verstraete et al.~\cite{verstraete2006criticality} and Somma et al.~\cite{somma2007quantum}. Let $H$ be a $k$-local classical Hamiltonian on $n$ qubits, with $T$ terms, where each qubit is acted on by at most $k'$ terms.

We can write the coherent Gibbs state $|\psi \rangle$ as
$$ |\psi \rangle = \frac{1}{\sqrt{Z}} e^{-H/2} |+ \rangle $$

Let $X_j$ be the Pauli $X$-matrix on qubit $j$. Using the representation above, one can check that:
$$ X_j |\psi \rangle = \Gamma_j |\psi \rangle \ , \ \Gamma_j = X_j e^{-H/2} X_j e^{H/2} $$
for each $j = 1, \dots, n$.

Note that the operator $\Gamma_j$ is diagonal in the computational basis. Since all matrix elements of $\Gamma_j$ are real, we conclude that $\Gamma_j$ is Hermitian. Note also that $\Gamma_j$ acts non-trivially only on $k'(k-1)+1$ qubits. Define the Hamiltonian
$$ H_{\text{SFF}} = \sum_j (\Gamma_j - X_j) $$

Note that $H_{\text{SFF}}$ is stoquastic. We have $H_{\text{SFF}} |\psi \rangle = 0$. The Perron-Frobenis theorem implies that $|\psi \rangle$ is the unique ground state of $H_{\text{SFF}}$. The same argument shows that $|\psi \rangle$ is the ground state of every local term $\Gamma_j - X_j$. Thus $H_{\text{SFF}}$ is a $(k'(k-1)+1)$-local SFF Hamiltonian with unique ground state $|\psi \rangle$.

\bigskip

\section{DBMs can represent HBMs}

We have seen from Theorems \ref{stoq_HBM} and \ref{SFF_HBM} that stoquastic and SFF ground states can be represented by HBMs. From Observation \ref{observation}, we also know that a classical thermal distribution is easily viewed as a HBM. Theorems \ref{stoq_DBM}, \ref{SFF_DBM} and \ref{Gibbs_RBM} are concerned with the representation of these distributions by a DBM. Thus to deduce Theorems \ref{stoq_DBM}, \ref{SFF_DBM} and \ref{Gibbs_RBM}, it is sufficient to show that any HBM can be represented by a DBM.

\begin{theorem}\label{HBM_DBM}
We are given a $k$-local HBM with $n$ visible nodes, $m$ hidden nodes, $T$ hyperedges, with each node contained in at most $k'$ hyperedges. We can find a DBM which represents the HBM in distribution to arbitrary precision*, with $n$ visible nodes, at most $2^k T$ hidden nodes in the middle layer, $m$ hidden nodes in the deep layer, and where each node has at most $k' 2^k$ connections.
\end{theorem}

Our strategy to prove Theorem \ref{HBM_DBM} is to find a RBM representing each hyperedge individually, and combine these to create a DBM representing the complete HBM. We adapt the following Lemma of Le Roux and Bengio\cite{le2008representational}:

\begin{lemma}\label{lerouxbengio}
Consider a distribution $\pi$ over $\{0,1\}^k$ with $\min_x (\pi(x)) > 0$. This distribution can be represented to arbitrary precision* (in total variation) by a RBM with $2^k$ hidden nodes.
\end{lemma}
\begin{proof}
Let 
\begin{align*}
&\lambda = \min_x (\pi(x))\\
&R = \max_x (\pi(x)) / \min_x (\pi(x))
\end{align*} 
Say we want to represent $\pi$ to precision $\epsilon$. Note by normalisation of $\pi$ we have $\lambda \leq 1/2^k$. We will construct an RBM with output $f$ such that $f$ is close to $\frac{1}{\lambda} \pi$. Suppose we have $|| f - \frac{1}{\lambda} \pi ||_{\text{TV}} \leq \epsilon'$. Then by Lemma \ref{norms} we have:
$$ \left| \left| \frac{f}{\sum_y f(y)} - \pi \right| \right|_{\text{TV}} \leq 2 \lambda \epsilon' $$

Thus we require $\epsilon' \leq \frac{\epsilon}{2 \lambda}$, and it is sufficient to take $\epsilon' = 2^{k-1} \epsilon$

Now we construct the RBM. Order $\{0,1\}^k = \{ x_1, \dots, x_{2^k} \}$ such that $\lambda = \pi(x_1) \leq \dots \leq \pi(x_{2^k}) = \lambda R$. Let $j$ be the lowest $i$ such that $\pi(x_i)/\lambda \geq 1 + \frac{\epsilon'}{2^k}$. Begin with the empty RBM, $f_{\text{empty}} (x) = 1 \ \forall x$. Let $a$ be a large real number. Now for each $i = j, \dots, 2^k$, add a hidden node with weight vector $w_i = a ( x_i - \frac{1}{2})$ and bias $c_i = - w_i^T x_i + \log \left( \frac{1}{\lambda} \pi(x_i) - 1 \right)$. The output of the resulting RBM is
$$ f(x_l) = \prod_{i \geq j} \left( 1 + e^{w_i^T x_l + c_i} \right) \ \text{if} \ l<j $$
$$ f(x_l) = \frac{1}{\lambda} \pi(x_l) \prod_{l \neq i \geq j} \left( 1 + e^{w_i^T x_l + c_i} \right) \ \text{if} \ l \geq j $$

For $l \neq i$, $w_i^T x_l + c_i \leq \log \left( \frac{\pi(x_i)}{\lambda} - 1 \right) - \frac{1}{2} a$. Thus we have
$$ 1 \leq f(x_l) \leq \left( 1 + R e^{- \frac{1}{2} a} \right)^{2^k} = 1 + O(2^k R e^{-\frac{1}{2}a}) \ \text{if} \ l<j $$
\begin{align*}
\frac{1}{\lambda} \pi(x_l) \leq f(x_l) &\leq \frac{1}{\lambda} \pi(x_l) \left( 1 + R e^{- \frac{1}{2} a} \right)^{2^k}\\& =\frac{1}{\lambda} \pi(x_l) + O(2^k R^2 e^{-\frac{1}{2}a}) \ \text{if} \ l \geq j 
\end{align*}

We can ensure $|| f - \frac{1}{\lambda} \pi ||_{\text{TV}} \leq 2^{k-1} \epsilon$ by taking 
$e^{-\frac{1}{2} a} = O(2^{-k} R^{-2} \epsilon)$ ie $\frac{1}{2} a = \log(O(2^k R^2 \epsilon^{-1}))$.
\end{proof}

{\bf Remark*}: The magnitude of the weights and biases of the RBM are bounded by $\log(O(2^k R^2 \epsilon^{-1})) = \poly(n)$ in all applications of Lemma \ref{lerouxbengio} in this paper.

Given this we can prove Theorem \ref{HBM_DBM} as follows.

\begin{proof}
Let the output of the HBM be $f(x)$. Consider the hyperedge $F_e$ on the $k$ nodes $(x_e, h_e)$ (the $x_e$ or $h_e$ variables possibly empty). Define the distribution
$$ \pi_e (x_e, h_e) = \frac{e^{-F_e (x_e, h_e)}}{\sum_{x_e,h_e} e^{-F_e (x_e, h_e)}} $$

Note that $f(x) \propto \sum_h \prod_e \pi_e (x_e,h_e)$, so 
$$ \frac{f(x)}{\sum_x f(x)} = \frac{\sum_h \prod_e \pi_e (x_e,h_e)}{\sum_{x,h} \prod_e \pi_e (x_e,h_e)} $$

We will construct the desired DBM from the HBM. Copy the $n$ visible nodes into the visible layer of the DBM, and the $m$ hidden nodes into the deep layer of the DBM. For a given hyperedge $e$, use Lemma \ref{lerouxbengio} to find a RBM (say $\text{RBM}_e$) representing the distribution $\pi_e$ to precision $\epsilon$ (in 1 norm). $\text{RBM}_e$ will have at most $2^k$ hidden nodes. Place these hidden nodes in the middle layer, copying the connections and biases from $\text{RBM}_e$. We do this for each hyperedge $e$.

Let the DBM have output $\tilde{f}$, and $\text{RBM}_e$ have output $\tilde{f}_e$. By construction, $\tilde{f}_e(x_e,h_e) \propto \pi_e(x_e,h_e) + O(\epsilon)$.
$$ \tilde{f}(x) = \sum_h \prod_e \tilde{f}_e (x_e, h_e) = Z \sum_h \prod_e (\pi_e(x_e,h_e) + O(\epsilon)) $$
for some $Z$.

Let $\min_{e,x_e,h_e} \pi_e(x_e,h_e) = \lambda$.
\begin{align*}
\tilde{f}(x) &= Z \sum_h \prod_e (\pi_e (x_e,h_e) + O(\epsilon)) \\&= Z \sum_h \prod_e (1 + O(\lambda^{-1} \epsilon)) \pi_e (x_e,h_e)\\
&= Z \sum_h (1 + O(T \lambda^{-1} \epsilon)) \prod_e \pi_e (x_e,h_e) \\&= (1 + O(T \lambda^{-1} \epsilon)) Z \sum_h \prod_e \pi_e (x_e,h_e)\\
\end{align*}
\begin{align*}
\left| \left| \tilde{f}(x) - Z \sum_h \prod_e \pi_e (x_e,h_e) \right| \right|_1 \\= O(T \lambda^{-1} \epsilon) Z \sum_{x,h} \prod_e \pi_e (x_e,h_e)
\end{align*}

By Lemma \ref{norms},
$$ \left| \left| \frac{\tilde{f}}{\sum_y \tilde{f}(y)} - \frac{f}{\sum_y f(y)} \right| \right|_{\text{TV}} = O(T \lambda^{-1} \epsilon) $$

Thus if we want the overall error to be $O(\delta)$, we require $\epsilon = O(T^{-1} \lambda \delta)$. If we examine the construction, we can see that the resulting DBM has at most $2^k T$ hidden nodes in the middle layer and $m$ hidden nodes in the deep layer, and each node has at most $k' 2^k$ connections.
\end{proof}

\begin{figure}[h]
\centering
  \includegraphics[width=0.5\textwidth]{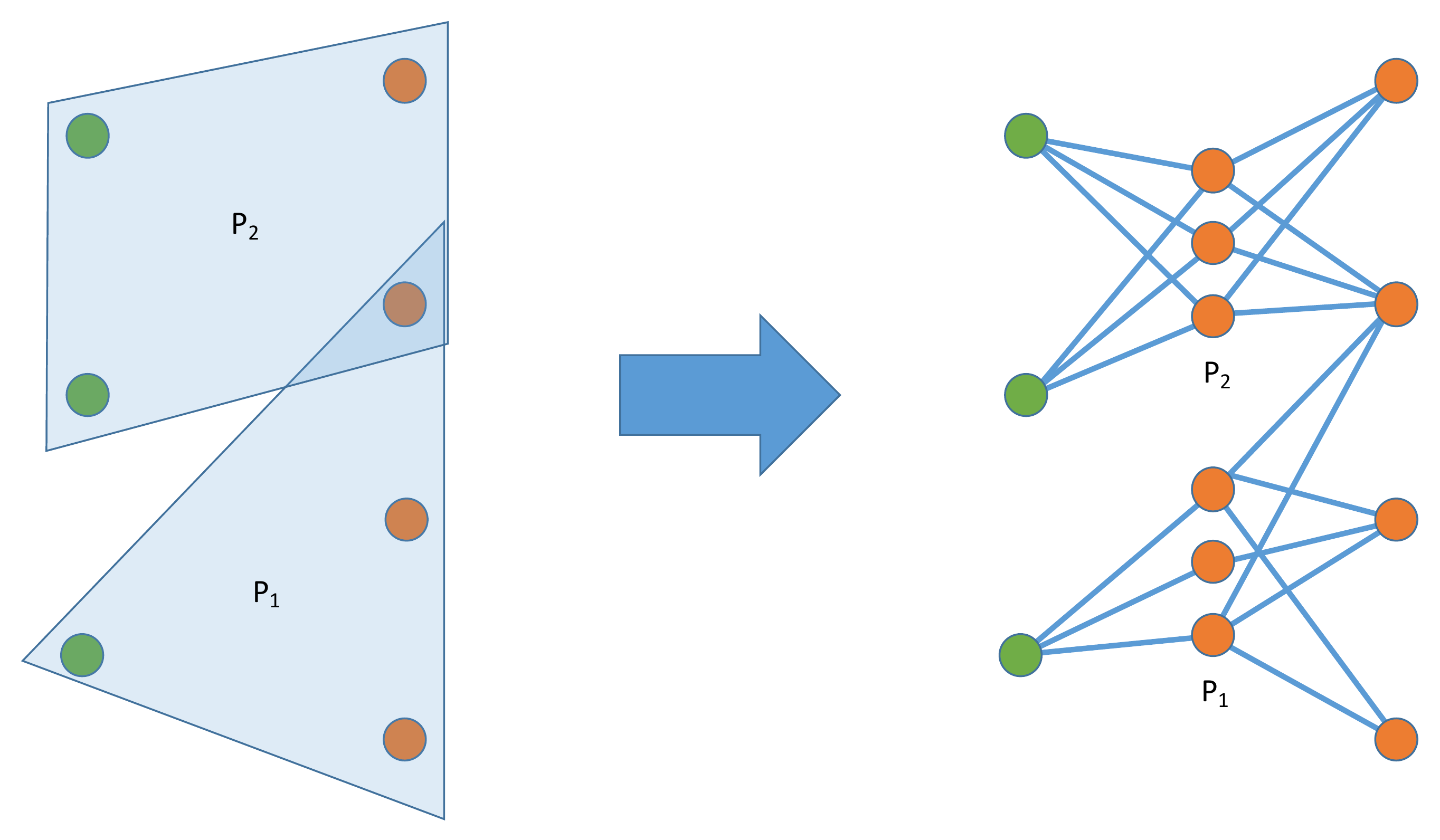}
  \caption{An example of converting two hyperedges into two local RBMs, leading to a global DBM.}
\end{figure}

As discussed above: Theorem \ref{HBM_DBM} and Theorem \ref{stoq_HBM} prove Theorem \ref{stoq_DBM}; Theorem \ref{HBM_DBM} and Theorem \ref{SFF_HBM} prove Theorem \ref{SFF_DBM}; and Theorem \ref{HBM_DBM} and Observation \ref{observation} prove Theorem \ref{Gibbs_RBM}. Another implication of this theorem is that any Boltzmann machine can be represented to arbitrary precision by a DBM with a polynomial overhead.

\bigskip

\section{Ising model}

The purpose of this section is to make explicit a relationship between the Boltzmann machine and the classical Ising model. First we define the Ising model: Consider a graph $(V,E)$, where each vertex $y$ is a spin $y \in \{-\frac{1}{2}, \frac{1}{2} \}$ with external field $a_y$, and each edge $\{y_1,y_2\} \in E$ carries a weight $W_{y_1,y_2}$. The (temperature 1) Ising model is the (temperature 1) Gibbs distribution of the classical Hamiltonian:
$$ H_{\text{Ising}} = - \sum_{y \in V} a_y y - \sum_{\{y_1,y_2\} \in E} W_{y_1,y_2} y_1 y_2 $$

Now consider a Boltzmann machine on the same graph and ignore the visible/hidden node distinction. The only difference between the Boltzmann machine and the Ising model is the values of the binary variables $y \in \{0,1\}$ for the Boltzmann machine versus the spin variables $y \in \{-\frac{1}{2}, \frac{1}{2}\}$ for the Ising model. We can change the biases/external fields respectively to account for this difference, so that the Hamiltonian of the Ising model and the energy of the Boltzmann machine exactly coincide. Thus their distributions will also coincide. To reintroduce the concept of hidden nodes, we must marginalize in the Ising model over spin variables corresponding to hidden nodes.

Thus we see that the Ising model and the Boltzmann machine are equivalent in the following sense: Any distribution represented by one can be represented by the other, as long as we allow ourselves to marginalize over a subset of spins in the Ising model.

\bigskip

\section{Classical sampling}\label{sampling}

\subsection{Gibbs sampling}

Representing a quantum state's distribution by a HBM provides a heuristic classical algorithm for sampling from the state, via Gibbs sampling. This is a special case of the Metropolis-Hastings algorithm. Say the energy of the HBM is $F(y) = \sum_e F_e(y)$, where $y = (x,h) \in \{0,1\}^N$, $N=n+m$, for $x \in \{0,1\}^n$ the visible nodes, and $h \in \{0,1\}^m$ the hidden nodes. Gibbs sampling sets up a Markov chain on the configuration space $y = \{0,1\}^N$, with each step requiring polynomial computation, whose stationary distribution is proportional to $\exp(-F(y))$. Running the Markov chain is then efficient, and the sample restricted to the visible nodes will converge to the HBM distribution, which is proportional to $\exp(-F(y))$ marginalised over the hidden nodes. However, it should be noted that in order to efficiently sample from the HBM distribution, we require the Gibbs sampling Markov chain to be fast mixing, which fails in some cases.

The Gibbs sampling procedure begins with a random configuration $y^{(0)} \in \{0,1\}^N$. At step $k$, we have $y^{(k-1)}$ and we wish to sample $y^{(k)}$. We sample each component $y_j^{(k)}$ of $y^{(k)} = (y_1^{(k)}, \dots y_N^{(k)})$ separately, starting with $j=1$. To sample $y_j^{(k)}$, we condition on the value of $(y_1^{(k)}, \dots, y_{j-1}^{(k)}, y_{j+1}^{(k-1)}, \dots y_N^{(k-1)})$. That is:
\begin{align*}
&\mathbb{P}(y_j^{(k)} | y_1^{(k)}, \dots, y_{j-1}^{(k)}, y_{j+1}^{(k-1)}, \dots y_N^{(k-1)}) \\
&= \frac{\mathbb{P}(y_1^{(k)}, \dots, y_{j-1}^{(k)}, y_j^{(k)}, y_{j+1}^{(k-1)}, \dots y_N^{(k-1)})}{\mathbb{P}(y_1^{(k)}, \dots, y_{j-1}^{(k)}, y_{j+1}^{(k-1)}, \dots y_N^{(k-1)})} \\
&= \frac{\exp(-F(y_1^{(k)}, \dots, y_{j-1}^{(k)}, y_j^{(k)}, y_{j+1}^{(k-1)}, \dots y_N^{(k-1)}))}{\sum_{u = 0,1} \exp(-F(y_1^{(k)}, \dots, y_{j-1}^{(k)}, u, y_{j+1}^{(k-1)}, \dots y_N^{(k-1)}))}
\end{align*}
$F$ is a sum of local terms, so these probabilities are efficiently computable. It can be checked that the detailed balance equations for this Markov chains are satisfied by $\exp(-F(y))$.

In the case of a DBM, we can streamline this process further. Say a 3 layer real DBM has visible layer $x$, middle hidden layer $h$, deep hidden layer $\tilde{h}$, and energy $F(x,h,\tilde{h}) = -a^T x - x^T W h - b^T h - h^T U \tilde{h} - c^T \tilde{h}$. The conditional distribution of a node conditional on the adjacent layer(s) takes a simple form:
\begin{align*}
\mathbb{P}(x_i = 1 | h) &= \sigma(a_i + W_{i \cdot} h) \\
\mathbb{P}(h_j = 1 | x, \tilde{h}) &= \sigma(x^T W_{\cdot j} + b_j + U_{j \cdot} \tilde{h}) \\
\mathbb{P}(\tilde{h}_k = 1 | h) &= \sigma(h^T U_{\cdot k} + c_k) \\
\end{align*}
where $\sigma(t) = 1/(1+\exp(-t))$. Note these probabilities are all efficiently computable. Now at step $k$ to sample $(x^{(k)}, h^{(k)}, \tilde{h}^{(k)})$, we can first sample $h^{(k)}$ conditional on $(x^{(k-1)}, \tilde{h}^{(k-1)})$, and then $(x^{(k)}, h^{(k)})$ conditional on $\tilde{h}^{(k)}$.

\subsection{Sampling using SFF Hamiltonian}

In \cite{bravyi2010complexity}, Bravyi and Terhal provide an algorithm for classical simulation of SFF ground states, based on a random walk on the basis states $\{0,1\}^n$. Consider a SFF Hamiltonian $H$ on $n$ qubits, and assume that the ground state $|\psi \rangle$ is unique, and has amplitudes $\langle x|\psi \rangle \geq 2^{- \poly(n)} \ \forall x$. (The situation is treated more generally in \cite{bravyi2010complexity}). Note that these assumptions are satisfied by the SFF Hamiltonian constructed in Section \ref{Bravyi_Terhal_SFF}. We can then set up a random walk on $\{0,1\}^n$ by specifying the probability of going from $x$ to $y$:
$$ \mathbb{P}(x \rightarrow y) = \frac{\langle y | \psi \rangle}{\langle x | \psi \rangle} \langle y|G|x \rangle $$
where $G = I - \beta H$, for some $\beta >0$ small enough so that $G$ has nonnegative entries. We start the walk at a random string in $\{0,1\}^n$. It can be checked that this is a well-defined Markov chain with stationary distribution $|\langle x|\psi \rangle|^2$ \cite{bravyi2010complexity}. It can also be shown that, with knowledge only of $H$ and not $|\psi \rangle$, these probabilities are efficiently computable and the Markov chain can be efficiently implemented \cite{bravyi2010complexity}. Moreover, it can be shown that the spectral gap of this Markov chain is equal to the spectral gap of $G$, which is $\beta \Delta$ where $\Delta$ is the spectral gap of $H$ \cite{bravyi2010complexity}. Thus if $H$ has an inverse polynomial gap, and has polynomial norm, then $\beta \Delta$ is inverse polynomial, so this Markov chain has a polynomial mixing time and the sampling algorithm becomes efficient.

Recall that in Section \ref{Bravyi_Terhal_SFF}, we took a local classical Hamiltonian $H_c$ and constructed a SFF Hamiltonian $H_{\text{SFF}}$ whose unique ground state is the coherent version of the Gibbs distribution of $H_c$. It is interesting to apply the above construction to $H_{\text{SFF}}$. The resulting random walk is as follows: for each $y$ which differs from $x$ on precisely one bit,
$$ \mathbb{P}(x \rightarrow y) = \beta \exp(- \frac{1}{2} (H_c(y) - H_c(x))) $$
With the remainder of the probability, remain at $x$. This walk can be efficiently implemented, and if $H_{\text{SFF}}$ has inverse polynomial gap, this walk has a polynomial mixing time, allowing efficient classical sampling of the Gibbs distribution of $H_c$.

\bigskip

\section{Discussion}
We showed the ground state space of stoquastic Hamiltonians can be efficiently described by classical thermal distributions which are represented by Hyper Boltzmann machines. This naturally leads to several interesting open questions.

In our work, we exhibit a partial equivalence between SFF ground states and classical thermal distributions (Theorems \ref{SFF_Gibbs} and \ref{Gibbs_SFF}). In order to map a SFF ground state to a classical thermal distribution, we require that the SFF Hamiltonian is gapped. Is it possible to remove this condition, and thus complete the equivalence between SFF ground states and classical thermal distributions? Similarly, for stoquastic Hamiltonians, in Theorem \ref{stoq_Gibbs} we require (a) the Hamiltonian to be gapped and (b) the ground state to have an inverse polynomial overlap with the uniform superposition state. Relaxing conditions (a) and (b) would complete the equivalence between stoquastic ground states and classical thermal distributions. We believe this would require a novel approach that differs from the Trotter decomposition of the imaginary time evolution operator which we use in this paper.

Another interesting direction is to investigate how the mappings in this paper relate to classical simulation of stoquastic ground states, and thus to the complexity theory of stoquastic Hamiltonians. In Section \ref{sampling}, we saw that if the ground state of a gapped SFF Hamiltonian has support on all the basis vectors, then it can be efficiently classically sampled. Moreover, it is often possible to efficiently classically sample from classical thermal distributions using techniques from Section \ref{sampling}. We pose this as another open problem: Are there conditions one can impose on the stoquastic Hamiltonian so that the related classical thermal distribution has an efficient sampling algorithm? If so, this mapping would provide a route to efficiently classically sampling from the stoquastic ground state.

{\it Acknowledgements}
Authors would like to thank Anurag Anshu for suggesting the correction for Lemma~\ref{SFF_P}.  
S.S. would like to thank Johannes Bausch and Joel Klassen for helpful discussions. 
S.S. acknowledges support from the QuantERA ERA-NET Cofund
in Quantum Technologies implemented within the European
Union’s Horizon 2020 Programme (QuantAlgo project), and
administered through EPSRC Grant No. EP/R043957/1, and
S.S. support from the Royal Society University Research Fellowship scheme.
\bigskip

\bibliographystyle{chicago}
\bibliography{biblio}

\bigskip

\section{Appendix}

\subsection{Technical lemma}

\begin{lemma}\label{norms}
Let $V$ be a normed space, $v,w \in V$, and $\hat{v} = v/||v||$, $\hat{w} = w/||w||$. If $||v - w|| \leq \epsilon$ then $|| \hat{v} - \hat{w} || \leq 2 \epsilon/||v||$.
\end{lemma}
\begin{proof}
$$ \left| ||v|| - ||w|| \right| \leq || v - w || \leq \epsilon $$
$$ || \hat{v} - \hat{w} || \leq \left| \left| \frac{v}{||v||} - \frac{w}{||v||} \right| \right| + \frac{1}{||v||} \left| ||w|| - ||v|| \right| \leq \frac{2\epsilon}{||v||} $$
\end{proof}

\subsection{Trotter decomposition}\label{Trotter}

Let $H = H_1 + \dots + H_L$, where each $H_i$ is bounded in by $|| H_i ||_{\text{op}} \leq J$. We want to show that:
\begin{align*}
\left| \left| e^{- \frac{1}{2} \delta H_L} \dots e^{- \frac{1}{2} \delta H_1} e^{- \frac{1}{2} \delta H_1} \dots e^{- \frac{1}{2} \delta H_L} - e^{-\delta H} \right| \right|_{\text{op}} = O(\delta^3 L^3 J^3)
\end{align*}

Let's expand both expressions.

\begin{align*}
&e^{- \delta H} =\\& 1 - \delta \sum_i H_i + \frac{\delta^2}{2} \sum_i H_i^2 + \frac{\delta^2}{2} \sum_{i < j} (H_i H_j + H_j H_i) + O(\delta^3 L^3 J^3) \\
&e^{- \frac{1}{2} \delta H_L} \dots e^{- \frac{1}{2} \delta H_1} e^{- \frac{1}{2} \delta H_1} \dots e^{- \frac{1}{2} \delta H_L} \\
&= \left(1 - \frac{\delta}{2} H_L + \frac{\delta^2}{8} H_L^2 + O(\delta^3 J^3) \right) \dots \\&\dots \left(1 - \frac{\delta}{2} H_1 + \frac{\delta^2}{8} H_1^2 + O(\delta^3 J^3) \right) \left(1 - \frac{\delta}{2} H_1 + \frac{\delta^2}{8} H_1^2 + O(\delta^3 J^3) \right) \dots \\&\dots \left(1 - \frac{\delta}{2} H_L + \frac{\delta^2}{8} H_L^2 + O(\delta^3 J^3) \right) \\
&= 1 - \delta \sum_i H_i + \frac{\delta^2}{2} \sum_i H_i^2 + \frac{\delta^2}{2} \sum_{i < j} (H_i H_j + H_j H_i) + O(\delta^3 L^3 J^3)
\end{align*}

Thus the two expressions differ by an error of $O(\delta^3 L^3 J^3)$.

\end{document}